\newif\ifChinese
\Chinesefalse
\documentclass[11pt,letterpaper]{article}
\usepackage[letterpaper,margin=1in]{geometry}
\usepackage{setspace}
\usepackage{amsmath,amsthm,amssymb,mathtools}
\usepackage[no-math]{fontspec}
\usepackage{microtype,booktabs,array}
\ifChinese
  \usepackage{xeCJK}
  \setCJKsansfont[Path=fonts/,Script=Default,BoldFont=FandolHei-Bold.otf]{FandolHei-Regular.otf}
  \setCJKmonofont[Path=fonts/,Script=Default]{FandolSong-Regular.otf}
  \renewcommand{\abstractname}{摘要}

\fi
\usepackage[numbers,sort&compress]{natbib}
\usepackage[colorlinks=true,citecolor=blue,linkcolor=blue,urlcolor=blue]{hyperref}
\hypersetup{pdfauthor={},pdfsubject={Tight value-oracle lower bound for monotone k-submodular maximization}}
\theoremstyle{plain}
\ifChinese
  \newtheorem{theorem}{定理}[section]
  \newtheorem{lemma}[theorem]{引理}
\else
  \newtheorem{theorem}{Theorem}[section]
  \newtheorem{lemma}[theorem]{Lemma}
\fi
\newcommand{\R}{\mathbb R}
\newcommand{\E}{\mathbb E}
\newcommand{\Prob}{\mathbb P}
\newcommand{\supp}{\operatorname{supp}}
\newcommand{\OPT}{\operatorname{OPT}}
\newcommand{\x}{\mathbf x}
\newcommand{\y}{\mathbf y}
\newcommand{\zero}{\mathbf 0}
\newcommand{\one}{\mathbf 1}
\newcommand{\pos}[1]{\left(#1\right)_+}
\makeatletter
\renewenvironment{abstract}{\normalsize\begin{center}\bfseries\abstractname\end{center}\quotation}{\endquotation}
\makeatother
\allowdisplaybreaks[1]
\date{}
\author{}

\usepackage{authblk}
\title{\textbf{A Tight Oracle Threshold for \(k\)-Submodular Maximization}}
\author[]{Siyuan Chen}
\affil[]{School of Mathematical Sciences,
	University of Chinese Academy of Sciences,
	Beijing 100049, China}
\date{\vspace{-1.5em}
	\small\texttt{chensiyuan221@mails.ucas.ac.cn}}
\hypersetup{pdftitle={A Tight Oracle Threshold for \(k\)-Submodular Maximization}}
\hypersetup{pdfauthor={Siyuan Chen}}

\begin{document}
\hypersetup{pageanchor=false}
\maketitle
\thispagestyle{empty}
\vspace{1.5em}

\begin{abstract}
We prove that, for  every \(\varepsilon>0\), a
\((\frac{k}{2k-1}+\varepsilon)\)-approximation algorithm for maximizing monotone
\(k\)-submodular functions requires exponentially many value-oracle queries.
This improves the \(\frac{k+1}{2k}\) oracle bound of Iwata, Tanigawa, and Yoshida.
Their \(\frac{k}{2k-1}\)-approximation algorithm therefore gives the exact threshold
in the monotone setting. For \(k\ge5\), our bound also matches, up to an
arbitrarily small loss, Aggarwal's recent approximation guarantee for
 non-monotone \(k\)-submodular maximization.
\end{abstract}

\section{Introduction}
Let \(V=\{e_1,\ldots,e_n\}\) be a ground set and \([k]=\{1,\ldots,k\}\).
Following Chen et al.~\cite{chen}, let \((k+1)^V\) denote the family of
\(k\)-tuples \((X_1,\ldots,X_k)\) of pairwise disjoint subsets of \(V\).
An assignment \(\x\in(k+1)^V\) gives each selected element a label in \([k]\).
Equivalently, \(\x_e=i\) if \(e\in X_i\), and \(\x_e=0\) if \(e\) is unassigned.
Set \(\supp(\x)=\bigcup_iX_i\) and \(\zero=(\emptyset,\ldots,\emptyset)\).
For \(\x=(X_1,\ldots,X_k)\) and \(\y=(Y_1,\ldots,Y_k)\), define
\[
 \x\sqcap\y=(X_i\cap Y_i)_{i=1}^k,
 \qquad
 \x\sqcup\y=\left((X_i\cup Y_i)\setminus\bigcup_{j\ne i}(X_j\cup Y_j)\right)_{i=1}^k.
\]
A function \(f:(k+1)^V\to\R\) is \emph{\(k\)-submodular} if
\begin{equation}\label{eq:ksub}
 f(\x)+f(\y)\ge f(\x\sqcap\y)+f(\x\sqcup\y)
 \qquad(\x,\y\in(k+1)^V).
\end{equation}
Write \(\x\preceq\y\) if \(X_i\subseteq Y_i\) for every \(i\).
The function is \emph{monotone} if \(f(\x)\le f(\y)\) whenever \(\x\preceq\y\).
For \(e\notin\supp(\x)\), its label-\(i\) marginal is
\[
 \Delta_{e,i}f(\x)=f(\x+i\one_e)-f(\x),
\]
where \(\x+i\one_e\) assigns label \(i\) to \(e\).
\emph{Orthant submodularity} means
\begin{equation}\label{eq:orthant}
 \Delta_{e,i}f(\x)\ge\Delta_{e,i}f(\y)
 \qquad(\x\preceq\y,\ e\notin\supp(\y),\ i\in[k]).
\end{equation}
When \(k=1\), \(k\)-submodularity is classical submodularity; in the unconstrained monotone
case, selecting the entire ground set \(V\) is optimal. We therefore consider \(k\ge2\).

In the value-oracle model, a query \(\x\) returns the exact value \(f(\x)\).
An \(\alpha\)-approximation outputs \(\x\) satisfying
\(\E[f(\x)]\ge\alpha\OPT(f)\), where
\(\OPT(f)=\max_{\y\in(k+1)^V}f(\y)\).
Iwata, Tanigawa, and Yoshida~\cite{ity} gave a
\(k/(2k-1)\)-approximation algorithm and proved that improving on
\((k+1)/(2k)\) requires exponentially many queries. We close this gap.

\begin{theorem}\label{thm:main}
For every fixed integer \(k\ge2\) and every \(\varepsilon>0\), any randomized
\((k/(2k-1)+\varepsilon)\)-approximation algorithm for maximizing nonnegative
monotone \(k\)-submodular functions requires exponentially many
value-oracle queries in the worst case.
\end{theorem}

Theorem~\ref{thm:main} matches the monotone algorithm of
Iwata et al.~\cite{ity}. Since the hard instances are monotone, the same
lower bound holds for the general nonnegative class. For \(k\ge5\),
Aggarwal~\cite[Corollary 4]{aggarwal} gives a
\((k/(2k-1)-\eta)\)-approximation for every \(\eta>0\).

We follow the hidden-partition, concentration, and adaptive-query
indistinguishability framework of Iwata et al.~\cite{ity}.
The main improvement is an explicit two-variable kernel with a one-sided
cutoff, obtained from Burkholder's homogeneous kernel~\cite{burkholder}.
The cutoff makes typical queries reveal no hidden-label information, while
the kernel preserves monotonicity and every orthant diminishing-return
inequality. 

\paragraph{Related work.}
Feige, Mirrokni, and Vondr\'ak~\cite{fmv} proved a \(1/2\) oracle
barrier even for symmetric nonnegative submodular functions, using a
hidden random partition and concentration. Vondr\'ak~\cite{vondrak}
developed the symmetry-gap framework, which converts a gap between
symmetric and optimal solutions into oracle hardness through refinement
and smoothing. Iwata et al.~\cite{ity} adapted the hidden-partition
approach to monotone \(k\)-submodular maximization.

\section{Construction and proof}
\paragraph{Overview.}
Give each element an independent hidden preferred label. For a query,
let \(x\) and \(y\) be the fractions of correctly and incorrectly labeled
elements. A fixed query has \(x\) close to \((x+y)/k\).
We construct a kernel that depends only on \(x+y\) throughout this region,
but takes a larger value at the fully correct assignment \((1,0)\).
We first verify the kernel by its derivatives, then prove that it induces
monotone \(k\)-submodular functions, and finally compare oracle transcripts.

\paragraph{The kernel.}
Fix \(0<\delta<(k-1)/k\). On
\(D=\{(x,y)\in\R_{\ge0}^2:x+y\le1\}\), put
\begin{equation}\label{eq:parameters}
 q=2-\frac1k,\qquad a=1-x-y,\qquad
 b=\pos{x-\frac{y+k\delta}{k-1}},\qquad A=a+b,
\end{equation}
where \((t)_+=\max\{t,0\}\), and define
\begin{equation}\label{eq:kernel}
 H(x,y)=1+\bigl((q-1)b-a\bigr)A^{q-1}
       =1-A^q+qbA^{q-1}.
\end{equation}
At \(A=0\), the product in~\eqref{eq:kernel} is zero.
The homogeneous term \(((q-1)b-a)(a+b)^{q-1}\) is a positive scalar
multiple of Burkholder's \(q\)-kernel~\cite[Eq.~(1.6)]{burkholder}.
Here the coordinate choice~\eqref{eq:parameters} and the one-sided cutoff
make that kernel suitable for the hidden-label construction.

\begin{lemma}\label{lem:kernel}
The function \(H\) is coordinatewise nondecreasing on \(D\), and each of
its first partial derivatives is coordinatewise nonincreasing.
Moreover,
\begin{equation}\label{eq:endpoint}
 H(0,0)=0,\qquad
 0\le H(x,y)\le H(1,0)
 =1+(q-1)\left(1-\frac{k\delta}{k-1}\right)^q.
\end{equation}
\end{lemma}
\begin{proof}
In the region \(b=0\), we have \(H=1-a^q\).
In the region \(b>0\),
\[
 b=x-\frac{y+k\delta}{k-1},\qquad
 A=1-\frac{k(y+\delta)}{k-1},\qquad
 (a_x,a_y)=(-1,-1),\quad
 (b_x,b_y)=\left(1,-\frac1{k-1}\right).
\]
Thus \(A_x=0\), \(A_y=-k/(k-1)=-1/(q-1)\), and differentiating
the second expression in~\eqref{eq:kernel} gives
\[
 H_x=qA^{q-1},\qquad
 H_y=qA^{q-1}-qbA^{q-2}=qaA^{q-2}.
\]
All first and second derivatives in the two regions are listed below;
the entries are taken away from the cutoff and the boundary \(x+y=1\).
\begin{center}
\renewcommand{\arraystretch}{1.35}
\begin{tabular}{@{}lcc@{}}
\toprule
 & \(b=0\) & \(b>0\)\\
\midrule
\(H_x\) & \(qa^{q-1}\) & \(qA^{q-1}\)\\
\(H_y\) & \(qa^{q-1}\) & \(qaA^{q-2}\)\\
\(H_{xx}\) & \(-q(q-1)a^{q-2}\) & \(0\)\\
\(H_{xy}=H_{yx}\) & \(-q(q-1)a^{q-2}\) & \(-qA^{q-2}\)\\
\(H_{yy}\) & \(-q(q-1)a^{q-2}\) &
\(-qA^{q-3}\left(\dfrac{k-2}{k-1}a+b\right)\)\\
\bottomrule
\end{tabular}
\end{center}
The first derivatives are nonnegative and the second derivatives are
nonpositive, since \(k\ge2\).

On the separating line \(x=(y+k\delta)/(k-1)\), we have \(b=0\)
and \(A=a\). The formulas from both regions therefore give
\(H_x=H_y=qa^{q-1}\), so the first derivatives agree across the line.

As \(A\to0\) within the region \(b>0\), the inequality \(0\le a\le A\)
gives
\[
0\le H_y=qaA^{q-2}\le qA^{q-1}=H_x\longrightarrow0.
\]
Within the region \(b=0\), we have \(A=a\) and
\[
H_x=H_y=qa^{q-1}\longrightarrow0
\qquad\text{as }A=a\to0.
\]
Thus both first derivatives extend continuously to \(D\).

Each horizontal or vertical segment in \(D\) crosses the separating
line at most once. On each side of the line, the signs of the second
derivatives in the table imply that both \(H_x\) and \(H_y\) are
nonincreasing along either coordinate direction. Their continuity
extends this property across the separating line and to the endpoints
of the segment. This proves the derivative assertions.

Monotonicity and \(H(0,0)=0\) give nonnegativity. Completing to the full
face yields \(H(x,y)\le H(1-y,y)\). On that face \(a=0\), so
\[
 H(1-y,y)=1+(q-1)\pos{1-\frac{k(y+\delta)}{k-1}}^q
 \le 1+(q-1)\left(1-\frac{k\delta}{k-1}\right)^q=H(1,0).
\]
This proves~\eqref{eq:endpoint}.
\end{proof}

\paragraph{The hidden functions.}
Choose \(\sigma_e\in[k]\) independently and uniformly for each \(e\in V\).
Equivalently, \(\sigma\) defines the hidden partition
\((\{e:\sigma_e=1\},\ldots,\{e:\sigma_e=k\})\).
For \(\x\in(k+1)^V\), set
\[
 m(\x)=|\supp(\x)|,\qquad
 C_\sigma(\x)=|\{e\in\supp(\x):\x_e=\sigma_e\}|,
\]
and define
\begin{equation}\label{eq:hidden}
 f_\sigma(\x)=
 \frac{H\bigl(C_\sigma(\x)/n,\,(m(\x)-C_\sigma(\x))/n\bigr)}{H(1,0)}.
\end{equation}

\begin{lemma}\label{lem:valid}
Every \(f_\sigma\) is nonnegative, monotone, and \(k\)-submodular,
with \(f_\sigma(\zero)=0\) and \(\OPT(f_\sigma)=1\).
\end{lemma}
\begin{proof}
	Assigning an unused element increases exactly one kernel coordinate
	by \(1/n\), according to whether its label equals \(\sigma_e\).
	Let \(\mathbf{r}_1=(1,0)\) and \(\mathbf{r}_2=(0,1)\), and write
	\(\partial_1H=H_x\) and \(\partial_2H=H_y\).
	For coordinatewise ordered points \(\mathbf{z}\le\mathbf{z}'\) in \(D\),
	an index \(j\in\{1,2\}\), and \(t\ge0\) with
	\(\mathbf{z}'+t\mathbf{r}_j\in D\), Lemma~\ref{lem:kernel} gives
	\begin{equation}\label{eq:finite}
		\begin{aligned}
			H(\mathbf{z}+t\mathbf{r}_j)-H(\mathbf{z})
			&=\int_0^t\partial_jH(\mathbf{z}+s\mathbf{r}_j)\,ds\\
			&\ge\int_0^t\partial_jH(\mathbf{z}'+s\mathbf{r}_j)\,ds\\
			&=H(\mathbf{z}'+t\mathbf{r}_j)-H(\mathbf{z}')\ge0.
		\end{aligned}
	\end{equation}
	Under an extension \(\x\preceq\y\), both kernel coordinates can only
	increase. Taking \(t=1/n\) in~\eqref{eq:finite} proves monotonicity and
	all inequalities~\eqref{eq:orthant} for \(f_\sigma\).
	
	For completeness, these properties imply~\eqref{eq:ksub} as follows.
	For compatible assignments \(\mathbf{u},\mathbf{v}\)
	(no element has two distinct nonzero labels), enumerate the elements of
	\(\supp(\mathbf{u})\setminus\supp(\mathbf{v})\).
	Add them with their labels first to \(\mathbf{u}\sqcap\mathbf{v}\)
	and then to \(\mathbf{v}\).
	At every step the first partial assignment precedes the second;
	summing~\eqref{eq:orthant} gives
	\[
	f_\sigma(\mathbf{u})-f_\sigma(\mathbf{u}\sqcap\mathbf{v})
	\ge f_\sigma(\mathbf{u}\sqcup\mathbf{v})-f_\sigma(\mathbf{v}).
	\]
	For arbitrary \(\x,\y\), delete every conflicting element from both,
	obtaining compatible \(\mathbf{u}\preceq\x\), \(\mathbf{v}\preceq\y\).
	Then \(\mathbf{u}\sqcap\mathbf{v}=\x\sqcap\y\) and
	\(\mathbf{u}\sqcup\mathbf{v}=\x\sqcup\y\), so
	\[
	f_\sigma(\x)+f_\sigma(\y)
	\ge f_\sigma(\mathbf{u})+f_\sigma(\mathbf{v})
	\ge f_\sigma(\x\sqcap\y)+f_\sigma(\x\sqcup\y).
	\]
	Finally, Lemma~\ref{lem:kernel} gives \(0\le f_\sigma\le1\) and
	\(f_\sigma(\zero)=0\). The assignment \(\x_e=\sigma_e\) has kernel
	coordinates \((1,0)\) and value \(1\).
\end{proof}

\paragraph{Indistinguishability and the value gap.}
We now prove Theorem~\ref{thm:main}. It suffices to consider
\(0<\varepsilon<1-1/q\), since a larger approximation guarantee implies
one with a smaller positive \(\varepsilon\). Set \(\delta=\varepsilon/4\).
For a fixed assignment \(\x\), independent of \(\sigma\), the random
variable \(C_\sigma(\x)\) is binomial with parameters \(m=m(\x)\) and \(1/k\).
The kernel lies in the region \(b>0\) exactly when
\(C_\sigma(\x)-m/k>\delta n\). For \(m>0\), Hoeffding's inequality gives
\begin{equation}\label{eq:concentration}
 \Prob_\sigma\!\left[C_\sigma(\x)-\frac{m}{k}>\delta n\right]
 \le e^{-2\delta^2n^2/m}\le e^{-2\delta^2n}.
\end{equation}
For \(m=0\), the event is empty. Outside this event, i.e, \(b=0\), the value is exactly
\begin{equation}\label{eq:fake}
 g(\x):=f_\sigma(\x)=\frac{1-(1-m(\x)/n)^q}{H(1,0)}\le\frac1{H(1,0)},
\end{equation}
independent of the hidden labels. The function \(g\) is itself a
nonnegative monotone \(k\)-submodular function: its nonnegative marginal
depends only on \(m\) and decreases with \(m\), and the argument of
Lemma~\ref{lem:valid} applies.

Suppose an algorithm makes at most \(Q\) queries on every input and
every run. Fix its random tape and run it against \(g\).
Its queries and final output are then fixed independently of \(\sigma\).
Applying~\eqref{eq:concentration} to these at most \(Q+1\) assignments
shows that all have \(b=0\), except with probability at most
\((Q+1)e^{-2\delta^2n}\). Whenever all these assignments have \(b=0\), the run against
\(f_\sigma\) has the same transcript: the first answers agree, and
equality of the previous answers forces the same next query and answer,
until the same output is returned. The output was included in the
probability bound, so its value is at most \(1/H(1,0)\).
On the exceptional event it is at most \(1\).
Averaging over the random tape therefore gives
\begin{equation}\label{eq:average}
 \E_{\sigma,\mathrm{alg}}[f_\sigma(\mathrm{output})]
 \le\frac1{H(1,0)}+(Q+1)e^{-2\delta^2n}.
\end{equation}

Since \(q-1=(k-1)/k\), convexity of \(t\mapsto(1-t)^q\) yields
\[
\begin{aligned}
	H(1,0)
	&=1+(q-1)\left(1-\frac{\delta}{q-1}\right)^q\\
	&\ge 1+(q-1)\left(1-\frac{q\delta}{q-1}\right)\\
	&=q-q\delta=q(1-\delta).
\end{aligned}
\]
Our choice gives \(0<\delta<1/8\) and \(q=2-1/k \ge 2-1/2=3/2\), hence
\[
 \frac1{H(1,0)}\le\frac1{q(1-\delta)}
 \le\frac1q+2\delta=\frac{k}{2k-1}+\frac\varepsilon2.
\]
If the algorithm guarantees \(k/(2k-1)+\varepsilon\) on every
\(f_\sigma\), then \(\OPT(f_\sigma)=1\) and~\eqref{eq:average} imply
\[
 \frac\varepsilon2\le(Q+1)e^{-2\delta^2n},\qquad
 Q\ge\frac\varepsilon2\exp\left(\frac{\varepsilon^2n}{8}\right)-1.
\]
This proves Theorem~\ref{thm:main}.

\paragraph{AI assistance.}
GPT 6.1 Astra constructed the hard function in this note, using the
Burkholder kernel cited above.

\end{document}